\documentclass{article}
\usepackage{amsmath,amssymb,latexsym,stmaryrd,amsthm,graphicx,color,framed}
\newtheorem{lemma}{\noindent Lemma}

\newcommand{\tuple}[1]{{\mbox{$\left\langle#1\right\rangle$}}}

\newcommand{\ignore}[1]{}

\def\unnamed#1#2{
\parbox{2.0in}{
\begin{tabbing}
\hspace{1em}\= #1 \\ \> \parbox{.75in}{\noindent \hrule ~} #2
\end{tabbing}
}}

\def\ant#1{\\ \> $#1$}

\newcommand{\intype}[2]{#1\!:\!#2}
\newcommand{\infc}{\!:\!}
\newcommand{\infcolon}{\!:\!}

\newcommand{\bool}{\mathrm{\bf Bool}}

\newcommand{\sett}{\mathbf{Set}}
\newcommand{\pair}{\mathbf{Pair}}
\newcommand{\id}{\mathbf{id}}

\newcommand{\refl}{\mathbf{refl}}

\newcommand{\group}{\mathbf{Group}}

\newcommand{\class}{\mathbf{Class}}

\newcommand{\double}[1]{\left\llbracket #1 \right\rrbracket}

\def\unnamed#1#2{
\parbox{2.0in}{
\begin{tabbing}
\hspace{1em}\= #1 \\ \> \parbox{.75in}{\noindent \hrule ~} #2
\end{tabbing}
}}

\newcommand{\true}{\mathbf{ True}}

\newcommand{\subst}{\mathbf{Subst}}
\newcommand{\point}{\mathbf{Point}}
\newcommand{\trans}{\mathbf{Trans}}
\newcommand{\predset}{\mathbf{PSet}}

\newcommand{\sigax}{\mathbf{SA}}

\newcommand{\Bi}{\mathbf{Bi}}

\newenvironment{infrule}{
  $\left\{\begin{array}{lrcl}}{
\end{array}\right.$}

\title{Dependent Type Theory as Related to the Bourbaki Notions of Structure and Isomorphism}

\author{David McAllester \\ \\ {\small Toyota Technological Institute at Chicago (TTIC)}}

\date{}

\begin{document}
\maketitle{}

\begin{abstract}
This paper develops a version of dependent type theory in which isomorphism is handled through
a direct generalization of the 1939 definitions of Bourbaki.  More specifically
we generalize the Bourbaki definition of structure from simple type signatures to dependent type signatures.
Both the original Bourbaki notion of isomorphism and its generalization given here
define an isomorphism between two structures $N$ and $N'$ to consist of bijections between their sorts
that transport the structure of $N$ to the structure of $N'$.  Here transport is defined
by commutativity conditions stated with set-theoretic equality.
This differs from the dependent type theoretic treatments of isomorphism given in the groupoid model and
homotopy type theory where no analogously straightforward set-theoretic definition of transport is specified.
The straightforward definition of transport also leads to a straightforward constructive
proof (constructive content) for the validity of the substitution of isomorphics --- something that is difficult
in the groupoid model or homotopy type theory.
\end{abstract}

\section{Introduction}

Isomorphism is a central tool in human mathematical reasoning.
We have a strong intuition that isomorphic objects are ``the same''.  For isomorphic graphs $G$ and $G'$, and a
graph-theoretic property $P$, we have $P(G)$ if and only if $P(G')$.  This is tautological if we define
``graph-theoretic'' to simply mean that this substitution property holds.  But it seems clear that being graph-theoretic is
actually a grammatical well-formedness condition.  A concept such as ``graph'' or
``topological space'' defines an interface to an object much as in object-oriented programming.  An object-oriented
compiler checks that a method or procedure defined on a certain class does not violate the abstract interface to that
class.  This is also possible for general formal mathematics as is demonstrated by the broad acceptance of formal verification systems based on
dependent type theory \cite{MathLib,LEAN}.  This paper explores the relationship between dependent type theory and isomorphism in the context of classical set theory.

The analysis of isomorphism in this paper involves a generalization of the definitions of Bourbaki \cite{Bourbaki39}.
The original definitions are based on simple types which we generalize here to dependent types.
The Bourbaki definition of structure can be paraphrased by saying that
a particular object has one or more carrier sets, a higher order (but simple) signature over those carrier sets specifying
constant, function and predicate symbols, and axioms that must be satisfied by any semantic interpretation of the symbols in the signature.
A group has one carrier set, the group elements, a signature specifying an identity constant, an inverse function, and the group operation,
and the group axioms.
We will call carrier sets sorts by analogy with the sorts of multi-sorted first order logic and first order structures.
A Bourbaki isomorphism between two structures $N$ and $N'$ is a system of bijections between their sorts that transports the structure
of $N$ to the structure of $N'$. Transport is determined by the signature independent of the axioms.
Furthermore the transport condition can always be written as a formula of set theory
without any recursive reference to the notion of isomorphism.

Isomorphism in dependent type theory has been treated in the groupoid model of Marten-L\"{o}ff type theory
\cite{GRPD} and in homotopy type theory \cite{Simpsets,Hott}.
In the groupoid model the sorts of a structure are taken to be groupoids --- categories in which every morphism is an isomorphism.
In Homotopy type theory the sorts are taken to be topological spaces and isomorphism is replaced by homotopy equivalence.
Unlike these previous treatments, here we
make and exploit the observation that every
class expression is equivalent (cryptomorphic) to a
Bourbaki structure definition generalized to dependent types. 
This Bourbaki approach has various advantages.
First, taking sorts to be unstructured sets rather than groupoids or topological spaces simplifies the formal treatment.
Second, the Bourbaki approach provides a direct definition of both isomorphism and transport using simple set-theoretic commutativity conditions
and results in a simple constructive proof (computational content)
for the theorem that isomorphism objects can be substituted in well-formed contexts.  Constructive proofs
for the substitution of isomorphics is much more awkward in the groupoid model or homotopy type theory.
Third, the Bourbaki treatment seems closer to the notion of isomorphism used in colloquial (rigorous but informal) mathematics.
We take the poition that the Bourbaki treatment is adequate for understanding the role of dependent type theory in the foundations of mathematics.
Finally, we also take the position that the Bourbaki treatment is adequate for the construction of automated reasoning systems that exploit
the notion of isomorphism and validity of the substitution of isomorphics.

The Bourbaki treatment presented here differs from Marten-L\"{o}ff type theory (MLTT) in the treatment of propositions.
In the type theory presented here propositions are Boolean valued expression as in classical logic.
In MLTT propositions are types and a proof of a proposition is an instance (or inhabitant) of the type.
The propositions-as-types treatment is motivated by constructivist philosophies of mathematics.
While a set-theoretic interpretation of dependent type theory
is straightforward, it is incompatible with propositions-as-types \cite{Reynolds84,MiquelWerner03}.
To bring type theory closer to classical logic there has been interest in ``proof irrelevant'' treatments of type theory.
Here propositions are still treated as types
but there is only one possible semantic inhabitant --- ``proved'' --- of any proposition type \cite{LeeWerner11,LEAN}.
But even proof-irrelevant type theories involve complexities arising from the residual treatment of propositions as types.
More specifically, it is nontrivial to allow $\mathbf{Prop}$ (the type of propositions) to be a valid sort.
Here we completely abandon propositions-as-types
and simply use formulas denoting Boolean values as in classical logic.  This yields a completely straightforward semantics.

Here we do not tie the semantics to any particular set of inference rules.  The semantics is just a Tarskian definition of meaning
as in semantic treatments of classical logic.  The semantics makes no reference to the $J$ operator of MLTT.
Implementations of formal verification
systems must be sound but are otherwise unconstrained --- they are not required to construct proofs over any given fixed set of inference rules.

In colloquial mathematics we have three notions of equality --- set-theoretic equality, isomorphism
and cryptomorphism.  Cryptomorphism, as discussed by Birkhoff \cite{Birk} and Rota \cite{Rota}, is an equivalence
between concept definitions.  For example a group can either be defined to be a set together with an identity
element, inverse operation and group operation satisfying certain axioms, or a set together with a group operation
such that an identity element and inverse elements exist, in which case they must be unique.
These are structurally different classes but we recognize
that these two definitions yield ``the same'' concept. Here we formulate a notion of functor based directly on type theory
independent of category theory.  A functor from class $\sigma$ to class $\tau$ is simply a term $F[x]\infc\tau$ with a free vaiable $x\infc\sigma$.
Because the language as a whole respects isomorphism, these functors respect isomorphism and can be forgetful. We can
define cryptomorphism in terms of this notion of functor ---  a cryptomorphism is a pair of functors between classes
establishing a bijection.

Colloquial mathematics also involves the notion of canonicality.
Every finite dimensional vector space is isomorphic to its dual
but there is no canonical isomorphism.  At a more elementary level there is no canonical basis for a vector space and no canonical point
on the surface of a sphere.  This corresponds to the fact that in a logic enforcing abstraction barriers and allowing the substitution of
isomorphics it is not possible to name (there is no well-typed expression denoting) an isomorphism between a vector space and its dual,
or a particular point on the surface
of a sphere.  There is no need for category theory in understanding cryptomorphism or canonicality.

Finally, we show that generalized Bourbaki isomorphism provides a model of $J$
and hence an interpretation of propositional equality in a variant of Marten-L\"{o}ff type theory.

\section{The Bourbaki Definitions}
\label{sec:Bourbaki}

We begin by reviewing the definition of structure and isomorphism given by the Bourbaki group of mathematicians.
A Bourbaki structure consists of a
``carrier set'' together with ``structure'' on that set.  For example, a directed graph can be defined to be a
set of nodes (the carrier set) and a edge predicat $E$ (the structure) where $E(n,m)$ is true for
nodes $n$ and $m$ if there is an edge from $n$ to $m$.  To handle the general case we let $\alpha$ be a variable ranging over the
carrier set. The structure provided by a directed graph can be specified by the type expression
$(\alpha \times \alpha) \rightarrow \bool$ which is the type of a
function taking two elements of $\alpha$ as arguments and returning a Boolean value.
In general a Bourbaki structure class has an associated type expression often called a signature.
The signature is a simple type expression over the carrier set variable $\alpha$.
A simple type expression over $\alpha$ is either the carrier variable $\alpha$, the constant $\bool$,
or an expression of the form $\sigma \times \tau$ or $\sigma \rightarrow \tau$
where $\sigma$ and $\tau$ are (recursively) simple
type expressions over $\alpha$.  For example a group has an identity element, an inverse operation and a group operation.
This corresponds to the signature $\alpha \times (\alpha \rightarrow \alpha) \times ((\alpha \times \alpha) \rightarrow \alpha)$.
A topological space is defined by a family of open sets and has signature $(\alpha \rightarrow \bool) \rightarrow \bool$.
For a simple type expression $\tau$ and a particular set $A$,
we will write $\tau[A]$ for the set denoted by the expression $\tau$ when the variable $\alpha$ is interpreted as the set $A$.
A structure with signature $\tau$ is a pair $\tuple{A,x}$ with $A$ a set and $x \in \tau[A]$.

A concept such as group or topological space is specified by giving a signature and axioms.
The axioms must respect the interface defined by the signature.
The isomorphism relation on a class is determined by the signature independent of the axioms.
For two structures $\tuple{A,x}$ and $\tuple{A',x'}$ with signature $\tau$,
Bourbaki defines an isomorphism from $\tuple{A,x}$ to $\tuple{A',x'}$ to be a bijection $f$ from $A$ to $A'$ that ``carries'' or ``transports''
$x$ to $x'$.
For sets $A$ and $A'$, a bijection $f$ from $A$ to $A'$, and a simple type expression $\tau[\alpha]$, we define the transport function $\intype{f_\tau}{\;(\tau[A] \rightarrow \tau[A'])}$ by

$$\begin{array}{cc}
  f_\alpha(x) = f(x) &   f_{\tau \times \sigma}(\tuple{x,y}) = \tuple{f_\tau(x),f_\sigma(y)} \\
  \\
  f_{\bool}(P)= P ~~ & ~~ f_{\tau \rightarrow \sigma}(g)(x) = f_\sigma(g(f_\tau^{-1}(x)))
  \end{array}$$

We then have that two structures $\tuple{A,x}$ and $\tuple{A',x'}$ of signature $\tau$ are isomorphic if there exists
a bijection $f$ from $A$ to $A'$ such that $f_\tau(x) = x'$.

\section{Bourbaki Type Theory}
\label{sec:BTT}

In this section we introduce a dependent type theory built on classical set theory.
The constructs of this type theory are essentially the same as those of Martin L\"{o}ff type theory (MLTT) \cite{mltt1973,mltt1990}
but without propositions-as-types or axiom J.  The language is specified semantically rather than proof theoretically.
The inference rules are implicitly any rules that are sound under the semantics.  This is analogous to the standard practice in mathematics of
defining notation by specifying what the notation means rather than specifying some semantically ambiguous syntactic calculus.

Following standard practice in the set theory, we assume the universe $V$ of sets.  All the expressions
of the formal language defined here denote either elements of $V$
or classes over $V$ (subsets of $V$ that are too large to be sets).\footnote{Section~\ref{sec:macros} describes
  a macro language supporting
  expressions that macro-expand to expressions denoting elements of $V$ or class-sized subsets of $V$.}
This is similar to ZFC set theory where all variables range only over elements of $V$ and proper classes are represented by formulas $\Phi[x]$.

Functions between proper classes (functors), such as the mapping from a topological manifold to its fundamental group,
can be represented by terms $G[x]$ where $x$ is a free variable of $G[x]$.
For example, we might have that for a variable $x$ ranging over topological spaces we have that $G[x]$ denotes a group.
The type system is motivated by the desire for expressions to respect isomorphism.
For example, if for any topological space $x$ we have that $G[x]$ denotes a group, and $X$ and $Y$ denote isomorphic topological spaces,
then $G[X]$ and $G[Y]$ should denote isomorphic groups.

Here we will work with tagged values --- pairs of a tag and contents --- where each tag is one of the tags ``boolean'', ``pair'', ``function'', ``set'' or ``class''.
A predicate $\intype{P}{s \rightarrow \bool}$ is tagged as a function and is different from the subset of $s$ satisfying $P$ which is tagged as a set.
We distinguish classes from sets by their tags --- a set is not a special case of a class.  A class expression will always
be either empty or denote a proper class (a collection too large to be a set).
We let $U_0$ be an alternate notation for $\sett$ and let $U_1$ be an alternate notation for $\class$.

\subsection{The Constructs of the Language}
\label{sec:constructs}

The following clauses give a somewhat informal definition of the semantics of the constructs.
In the following we write $\intype{e}{\sigma}$ to mean that $e$ denotes an element of the set or class $\sigma$.

\begin{itemize}
\item We have $\intype{{\color{red} \sett}}{\class}$, or equivalently $\intype{U_0}{U_1}$, where $\sett$ denotes the class of all sets.

\item We have $\intype{{\color{red} \bool}}{\sett}$ where $\bool$ denotes the set containing the two
  truth values where truth values are tagged as Boolean.

\item For $\sigma:\!U_i$ and $\tau[x]\!:\!U_j$ for all $x \in \sigma$ we have $({\color{red} \Sigma_{\intype{x\;}{\;\sigma}}\;\tau[x]})\!:\!U_{\max(i,j)}$
  where $\Sigma_{\intype{x\;}{\;\sigma}}\;\tau[x]$ denotes the
  set or class of pairs $\tuple{x,y}$ with $x \in \sigma$ and $y \in \tau[x]$.   If $x$ does not occur in $\tau[x]$
  then we write $\Sigma_{\intype{x\;}{\;\sigma}}\;\tau[x]$ as {\color{red} $\sigma \times \tau$}.

\item For $\sigma:\!U_i$ and $\tau[x]\!:\!U_j$ for all $x \in \sigma$
  and for $u\!:\!\sigma$ and $w\!:\!\tau[u]$ we have ${\color{red} \tuple{u,w}}\!:\!\Sigma_{\intype{x\;}{\;\sigma}}\;\tau[x]$ where $\tuple{u,w}$ denotes the pair of $u$ and $w$.

\item For $e\!:\!\Sigma_{\intype{x\;}{\;\sigma}}\;\tau[x]$
  the projections ${\color{red} \pi_1(e)}\!:\!\sigma$ and ${\color{red} \pi_2(e)}\!:\!\tau[\pi_1(e)]$ denote the first and second components the pair.

\item For $s\!:\!\sett$ and $u[x]\!:\!\sett$ for all $x \in s$, we have
  $\intype{({\color{red} \Pi_{\intype{x\;}{\;s}}\;u[x]})}{\sett}$
  where $\Pi_{\intype{x\;}{\;s}}\;u[x]$ denotes the set of all
  functions $f$ with domain $s$ and such that
  for all $x \in s$ we have that $f(x)\!:\!u[x]$. If $x$ does
  not occur in $u[x]$ we abbreviate
  $\Pi_{\intype{x\;}{\;s}}\;u[x]$ as {\color{red} $s \rightarrow u$}.

\item For $s\!:\!\sett$, $u[x]\!:\!\sett$ and $e[x]\!:\!\tau[x]$ for all $x \in s$, we have
  $({\color{red} \lambda\;\intype{x}{s}\; e[x]})\!:\!\Pi_{\intype{x\;}{\;s}}\;e[x]$
  where $\lambda\;\intype{x}{s}\; e[x]$ denotes the function mapping an element $x\!:\!s$ to the value $e[x]$.

\item For $f\!:\!\Pi_{\intype{x\;}{\;s}}\;u[x]$ and $e\!:\!s$ we have ${\color{red} f(e)}\!:\!u[e]$
 where $f(e)$  denotes the value of the function $f$ on argument $e$.
 
\item For $\sigma\!:\!U_i$  and $\Phi[x]\!:\!\bool$ for all $x \in \sigma$,
  we have $({\color{red} S_{\intype{x\;}{\;\sigma}}\;\Phi[x]})\!:\!U_i$ where $S_{\intype{x\;}{\;\sigma}}\;\Phi[x]$ denotes the
  subset of $\sigma$ consisting of all $x\!:\!\sigma$ such that $\Phi[x]$ is true.

\item For $\intype{u}{s}$ and $\intype{w}{s}$ with $\intype{s}{\sett}$ we have $\intype{({\color{red} u = w})}{\bool}$
  where $u = w$ is true if $u$ is the same as $w$ (set-theoretic equality).

\item For $\intype{N}{\sigma}$ and $\intype{M}{\sigma}$ with $\intype{\sigma}{\class}$ we have $\intype{({\color{red} N =_\sigma M})}{\bool}$
  where $N =_\sigma M$ is true if $N$ is $\sigma$-isomorphic to $M$ as defined in section~\ref{sec:isomorphism}.

\item For $\sigma\!:\!U_i$ and $\Phi[x]\!:\!\bool$ for all $x \in \sigma$,
  we have $({\color{red} \forall\intype{x}{\sigma}\;\Phi[x]})\!:\!\bool$ where $\forall\intype{x}{\sigma}\;\Phi[x]$ is true if for every element $x$ of $\sigma$
  we have that $\Phi[x]$ is true.  The formula ${\color{red} \exists\intype{x}{\sigma}\;\Phi[x]}$ is defined similarly.

\item The Boolean formulas {\color{red} $\neg \Phi$, $\Phi \vee \Psi$, $\Phi  \Rightarrow \Psi$}
  and ${\color{red} \Phi \Leftrightarrow \Psi}$ have their
  classical Boolean meaning.  For example, if $\intype{\Phi}{\bool}$ and $\intype{\Psi}{\bool}$ then $\intype{(\Phi \Rightarrow \Psi)}{\bool}$
  where $\Phi \Rightarrow \Psi$ is true unless $\Phi$ is true and $\Psi$ is false.

\item If there exists exactly one $\intype{x}{s}$ such that $\Phi[x]$ with $\intype{s}{\sett}$ then
  $\intype{({\color{red} \mathbf{The}\;\intype{x}{s}\;\Phi[x]})}{s}$ denotes that $x$.
\end{itemize}

As an example the class of groups can be defined as
\begin{eqnarray*}
  \mathbf{GroupSig} & \equiv & \Sigma_{\intype{\alpha\;}{\;\sett}}\;\alpha \times (\alpha \rightarrow \alpha) \times ((\alpha \times \alpha) \rightarrow \alpha) \\
  \\
  \mathbf{Group} & \equiv & S_{\intype{G\;}{\;\mathbf{GroupSig}}}\;\Phi(G)
\end{eqnarray*}
where $\Phi(G)$ states the group axioms.

Well-formedness is relative to a context declaring the types for variables and stating assumptions.
A context $\Gamma$ consists of variable declarations $\intype{x}{\tau}$ and assumptions $\Phi$. Contexts themselves are subject to
well-formedness constraints. We write $\Gamma \models \intype{e}{\sigma}$ to mean that $\Gamma$ is well-formed, that both $e$ and $\sigma$ are well-formed under $\Gamma$,
and that for all variable interpretations satisfying the context
$\Gamma$ we have that the value of $e$ is a member of the value of $\sigma$. The empty context is well-formed and we have $\models \intype{\sett}{\class}$.
For $\Gamma$ a well-formed context and $\Gamma \models \intype{\sigma}{U_i}$ we have that $\Gamma;\;\intype{x}{\sigma}$ is well-formed for any variable $x$
not previously declared in $\Gamma$.
For $\Gamma \models \intype{u}{s}$, $\Gamma \models \intype{w}{s}$ with $\Gamma \models \intype{s}{\sett}$, we have $\Gamma \models \intype{u = w}{\bool}$ where
$u = w$ is true if $u$ equals $v$ (set-theoretic equality) under all variable interpretations satisfying $\Gamma$.
For $\Gamma \models \intype{\Phi}{\bool}$ we write $\Gamma \models \Phi$
to mean that $\Phi$ is true under all variable interpretations satisfying $\Gamma$.
Finally, for $\Gamma \models \intype{\Phi}{\bool}$ we have that $\Gamma;\Phi$ is a well-formed context.

A variable $\alpha$ declared by $\intype{\alpha}{\sett}$ can be viewed as a ``sort'' in the sense of a multi-sorted first order signature.
Sorts are fundamental to isomorphism.  The type of an element of a sort is completely undetermined and we can think of the elements
of a sort as structureless ``points''.  Points are discussed more formally in section \ref{sec:representation}.  The type expression
$$\Sigma_{\intype{\alpha_1\;}{\;\sett}}\;\Sigma_{\intype{\alpha_2\;}{\;\sett}}\;\alpha_1\rightarrow \alpha_2$$
defines a signature with two sorts $\alpha_1$ and $\alpha_2$.
It seems natural that isomorphism between multi-sorted structures is a system of
bijections between the corresponding sorts.

A sequent $\Gamma \models \intype{\Phi}{\bool}$ states that $\Phi$ will have the same truth value under $\Gamma$-isomorphic variable interpretations
where we view $\Gamma$ as defining a class and view a variable interpretation satisfying $\Gamma$ as a structure in that class.
Care must be taken to restrict equality formation to equalities that respect isomorphism.
The well-formedness condition on equality can be characterized by the following ``inference rule'' where the antecedents imply the conclusion.

\vspace{-3ex}
~ \hspace{8em} \unnamed{
  \ant{\Gamma \models \intype{s}{\sett}}
  \ant{\Gamma \models \intype{x}{s}}
  \ant{\Gamma \models \intype{y}{s}}
}{
  \ant{\Gamma \models \intype{(x=y)}{\bool}}
}

\noindent We have
$$\intype{\alpha}{\sett};\;\intype{x}{\alpha};\;\intype{y}{\alpha} \models \intype{(x=y)}{\bool}.$$
However, we cannot form set-theoretic equalities between elements of different sorts.
$$\intype{\alpha_1}{\sett};\;\intype{\alpha_2}{\sett};\;\intype{x}{\alpha_1};\;\intype{y}{\alpha_2} \not \models \intype{(x=y)}{\bool}.$$
But we have
$$\intype{\alpha}{\sett};\;\intype{f}{\alpha\rightarrow \alpha};\;\intype{g}{\alpha\rightarrow \alpha} \models \intype{(f=g)}{\bool}.$$
and
$$\intype{\alpha_1}{\sett};\;\intype{\alpha_2}{\sett};\;\intype{f}{\alpha_1\rightarrow \alpha_2};\;\intype{g}{\alpha_1\rightarrow \alpha_2} \models \intype{(f=g)}{\bool}.$$

\noindent Set-theoretic equality between elements of classes is not well-formed and we have
\begin{eqnarray*}
\intype{\alpha_1}{\sett};\;\intype{\alpha_2}{\sett} & \not \models & \intype{(\alpha_1=\alpha_2)}{\bool} \\
\intype{G_1}{\mathbf{Group}};\;\intype{G_2}{\mathbf{Group}} & \not \models&  \intype{(G_1=G_2)}{\bool}
\end{eqnarray*}
\noindent But isomorphism between elements of classes is well formed and we have
\begin{eqnarray*}
  \intype{\alpha_1}{\sett};\;\intype{\alpha_2}{\sett} & \models & \intype{(\alpha_1=_{\sett} \alpha_2)}{\bool} \\
  \intype{G_1}{\mathbf{Group}};\;\intype{G_2}{\mathbf{Group}} & \models&  \intype{(G_1=_{\mathbf{Group}}G_2)}{\bool}
\end{eqnarray*}

\noindent We can write the following inference rule for the well-formedness of isomorphism equations.

\vspace{-3ex}
~ \hspace{8em} \unnamed{
  \ant{\Gamma \models \intype{\sigma}{\class}}
  \ant{\Gamma \models \intype{N}{\sigma}}
  \ant{\Gamma \models \intype{M}{\sigma}}
}{
  \ant{\Gamma \models \intype{(N=_\sigma M)}{\bool}}
}

\subsection{Minimum Well-Formedness}

Tarskian semantic value functions are typically defined by recursive descent into expressions.  This is a well-founded
recursion that defines a unique meaning.  The above clauses defining well-formedness and meaning are recursive but the recursion is not a simple descent into expressions.
The clauses assert that certain expressions are well-formed.  But they do not explicitly assert what is not well-formed.
We define the well-formed expressions to be only those that are required to be well-formed by the clauses ---
the minimum set of well-formed expressions satisfying the clauses.
This least fixed point semantics for well-formedness supports proofs
by ``induction on the formation of expressions'' which is used in section~\ref{sec:isomorphism}.

\subsection{Functors, Macros and Cryptomorphism}
\label{sec:macros}

Following the terminology of Birkoff \cite{Birk} and Rota \cite{Rota},
two classes $\sigma$ and $\tau$ that are well-formed under $\Gamma$
will be called cryptomorphic (in context $\Gamma$) if there exist
well-formed functor expressions $F[x]$ and $G[y]$
such that
$\Gamma;\;\intype{x}{\sigma} \models G[F[x]] = x$ and $\Gamma;\;\intype{y}{\tau} \models F[G[y]] = y$.
Here the equations are taken to be set-theoretic rather than expressing isomorphism.  The case where
the equalities are taken to be isomorphism is also interesting but will not be discussed here.

It will be convenient to write a functor $\intype{F}{\sigma \rightarrow \tau}$ as a lambda expression
$\lambda \intype{x}{\sigma}\;G[x]$.  However, this is viewed here as syntactic sugar.
For $F = \lambda \intype{x}{\sigma}\;G$ we have that $F[X]$ simply abbreviates
the result of substituting $X$ for $x$ in $G$.  Here functors are viewed as macros.
We will write applications of macros using square brackets such as $F[X]$ to emphasize that
this application represents a syntactic substitution rather than an application of a semantic function.

We can also allow higher order macros that can take macros as arguments and can return macros as values.
We define a macro type to be a set expression, a class expression, or a macro type of the form $\Pi_{x:M_1} \rightarrow M_2[x]$
where $M_1$ and $M_2$ are macro types and where we have the following inference rules for defining the well-formedness of macro expressions.

\vspace{-3ex}
~\hfill \unnamed{
  \ant{\Gamma;\;\intype{x}{M_1} \;\models\; \intype{G}{M_2[x]}}
}{
  \ant{\Gamma\;\;\models\;\intype{(\lambda\;\intype{x}{M_1}\;G)}{\Pi_{x:M_1} \; M_2[x]}}
}
\hfill
\unnamed{
  \ant{\Gamma \models \intype{F}{\Pi_{x:M_1} \; M_2[x]}}
  \ant{\Gamma \models \intype{G}{M_1}}
}{
  \ant{\Gamma \models \intype{F[G]}{M_2[G]}}
}
\hfill ~

Macros are lambda expressions under a term model semantics. In the presence of macros we have that
for $\intype{N}{\sigma}$, where $N$ and $\sigma$ can contain macro applications but not free macro variables,
the expression $N$ strongly normalizes under $\beta$-reduction to a well-formed expression $N'$ in the base language (the language without macros)
and $\sigma$ similarly normalizes to a well-formed base language type $\sigma'$ with $N' \infcolon \sigma'$.  We will show that base expressions respect isomorphism.

For a macro variable $P$ declared by, say, $\intype{P}{(\mathbf{Group} \rightarrow \bool)}$ it is
important that $P$ is ranging over well-formed lambda expressions and not arbitrary functions.
There exist set-theoretic predicates that distinguish isomorphic groups and hence fail to respect
isomorphism.  But if $P[G]$ $\beta$-reduces (via substitution) to a well-formed Boolean expression involving $G$ then
$G =_{\mathbf{Group}} \;G'$ implies $P[G] \Leftrightarrow P[G']$.

\subsection{Signature-Axiom Classes}
\label{sec:sigax}

We define a signature-axiom (SA) class expression to be a class expression of the form
$$\Sigma_{\intype{\alpha\;}{\;\sett^n}}\;S_{\intype{x\;}{\;s[\alpha]}}\;\Phi[\alpha,x]$$
where $\sett^n$ abbreviates $\sett \times \cdots \times \sett$ with $n$ occurances of $\sett$ and
where $s[\alpha]$ is a set expression.
We call $s[\alpha]$ the signature and $\Phi[\alpha,s]$ the axioms. Here $\alpha$ gives the list of sorts --- typically a single sort
but we allow for multi-sorted signatures.
The notion of isomorphism for a signature-axiom class depends on the signature but not on the axioms.

\noindent For example the class of groups can be written as
$$\Sigma_{\intype{\alpha\;}{\;\sett}}\;S_{\intype{x\;}{\;\alpha \times (\alpha \rightarrow \alpha) \times ((\alpha \times \alpha) \rightarrow \alpha)}} \;\Phi[\alpha,x]$$
where $\Phi[\alpha,s]$ states the group axioms.
The class of topological spaces can be written as
$$\Sigma_{\intype{\alpha\;}{\;\sett}}\;S_{\intype{\mathrm{Open}\;}{\;(\alpha \rightarrow \bool) \rightarrow \bool}}\;\Phi[\alpha,\mathrm{Open}]$$
where $\Phi[\alpha,\mathrm{Open}]$ states that axioms of point-set topology.

\medskip
\begin{lemma}
  Every class expression is cryptomorphic to a signature-axiom class.
\end{lemma}

\begin{proof}
For a given class expression $\sigma$ we will define a corresponding structure-axiom class $\overline{\sigma}$ and functors $\sigax_\sigma$
and $\sigax^{-1}_\sigma$ satisfying the following inference rules.

\noindent
\begin{infrule}
  \Gamma & & \models & \intype{N}{\sigma} \\
  \hline \\
  \Gamma & & \models & \intype{\sigax_\sigma[N]}{\overline{\sigma}} \\
  \Gamma & & \models & \sigax^{-1}_\sigma[\sigax_\sigma[N]] = N
\end{infrule}

\noindent \begin{infrule}
  \Gamma & & \models & \intype{S}{\overline{\sigma}} \\
  \hline \\
  \Gamma & & \models & \intype{\sigax^{-1}_\sigma[S]}{\sigma} \\
  \Gamma & & \models & \sigax_\sigma[\sigax^{-1}_\sigma[S]] = S
\end{infrule}

A class expression is either the constant $\sett$ or a class of pairs where each such pair contains a set within it. The functor $\sigax_\sigma(N)$ moves all sets to the front.
This is just a rearrangement of the pairing structure.  Note that reversing the two components of a pair is a
cryptomorphism between $\sigma \times \tau$ and $\tau \times \sigma$.  But note that for the class of pointed sets $\Sigma_{\intype{\alpha\;}{\;\sett}}\;\alpha$
it is important that the set (sort) $\alpha$ comes before the point of that sort --- if the point comes first it is not possible to give it a type.
There is no ``reversal cryptomorphism'' for pointed sets. But the sets (sorts) can always be moved to the top of the class.
Note that the elements of a signature-axiom class $\Sigma_{\intype{\alpha\;}{\;\sett^n}}\;S_{\intype{x\;}{\;s[\alpha]}}\;\Phi[\alpha,x]$ are pairs of the form $\tuple{A,x}$
where $x$ is a set element.  Set elements are either points (elements of a sort variable as discussed in section~\ref{sec:representation}), pairs of set elements, or functions between set elements.
This implies that set elements are never themselves sets. So the only sets (sorts) in an instance $\tuple{A,x}$ of a signature-axiom class are the sets (sorts) in $A$.

A class expression must be either the constant $\sett$, a dependent pair class $\Sigma_{\intype{x\;}{\;\sigma}}\;\tau[x]$, or a subclass
$S_{\intype{x\;}{\;\sigma}}\;\Phi(x)$ and we can define $\overline{\sigma}$, $\sigax_\sigma$ and $\sigax^{-1}_\sigma$ by structural induction on $\sigma$.
For the class $\sett$ we have
\begin{eqnarray*}
  \overline{\sett} & :\equiv & \Sigma_{\intype{\alpha\;}{\;\sett}}\;S_{\intype{P\;}{\;\bool}}\;P \\
  \sigax_{\sett}[s] & :\equiv & \pair(s,\true) \\
  \sigax^{-1}_{\sett}[p] & :\equiv & \pi_1(p)
\end{eqnarray*}
For $\sigma = S_{\intype{x\;}{\;\tau}}\;\Phi[x]$ we let
$\Sigma_{\intype{\alpha\;}{\;\sett^n}}\;S_{\intype{x\;}{\;s[\alpha]}}\;\Psi[\alpha,x]$
be $\overline{\tau}$
and define
\begin{eqnarray*}
  \overline{\sigma} & :\equiv & \Sigma_{\intype{\alpha\;}{\;\sett^n}}\;S_{\intype{x\;}{\;s[\alpha]}}\; \Psi[\alpha,x] \wedge \Phi[\sigax^{-1}_\tau[\pair(\alpha,x)]] \\
  \sigax_\sigma[N] & :\equiv & \sigax_\tau[N] \\
  \sigax^{-1}_\sigma[S] & :\equiv & \sigax^{-1}_\tau[S]
\end{eqnarray*}
For $\sigma = \Sigma_{\intype{x\;}{\;\tau}}\;\gamma[x]$ we have that one of $\tau$ and $\gamma[x]$ must be a class expression. If $\gamma[x]$ is a set expression
then $\tau$ must be a class expression in which case we let
$\Sigma_{\intype{\alpha\;}{\;\sett^n}}\;S_{\intype{x\;}{\;s[\alpha]}}\;\Phi[\alpha,x]$ be $\overline{\tau}$ and define
\begin{eqnarray*}
  \overline{\sigma} & :\equiv & \Sigma_{\intype{\alpha\;}{\;\sett^n}}\;S_{\intype{z\;}{\;\Sigma_{\intype{x\;}{\;s[\alpha]}}\;\gamma[\sigax^{-1}_\tau[\pair(\alpha,x)]]}}\; \Phi[\alpha,\pi_1(z)] \\
  \sigax_\sigma[N] & :\equiv & \pair(\pi_1(W),\pair(\pi_2(W),\pi_2(N)));\;\;\;W = \sigax_\tau[\pi_1(N)] \\
  \sigax_\sigma^{-1}[S] & :\equiv & \pair(\sigax^{-1}_\tau[\pair(\pi_1(s),\pi_1(\pi_2(S)))],\pi_2(\pi_2(S)))
\end{eqnarray*}
For $\sigma = \Sigma_{\intype{x\;}{\;\tau}}\;\gamma[x]$ with $\tau$ a set expression we must have that $\gamma[x]$ is a class expression in which case
we let
$\Sigma_{\intype{\alpha\;}{\;\sett^n}}\;S_{\intype{w\;}{\;s[x,\alpha]}}\;\Phi[x,\alpha,w]$ be $\overline{\gamma[x]}$ and define
\begin{eqnarray*}
  \overline{\sigma} & :\equiv & \Sigma_{\intype{\alpha\;}{\;\sett^n}}\;S_{\intype{z\;}{\;\Sigma_{\intype{x\;}{\;\tau}}\;s[x,\alpha]}}\; \Phi[\pi_1(z),\alpha,\pi_2(z)] \\
  \sigax_\sigma[N] & :\equiv & \pair(\pi_1(W),\pair(\pi_1(N),\pi_2(W)));\;\;\;W = \sigax_{\gamma[x]}[\pi_2(N)] \\
  \sigax_\sigma^{-1}[S] & :\equiv & \pair(\pi_1(\pi_2(S)),\sigax^{-1}_{\gamma[x]}[\pair(\pi_1(s),\pi_2(\pi_2(S)))])
\end{eqnarray*}

For $\sigma = \Sigma_{\intype{x\;}{\;\tau}}\;\gamma[x]$ with both $\tau$ and $\gamma[x]$ class expressions we let
\newline $\Sigma_{\intype{\alpha\;}{\;\sett^n}}\;S_{\intype{w\;}{\;s_1[\alpha]}}\;\Phi[\alpha,w] \;\;\mbox{and} \;\;\Sigma_{\intype{\beta\;}{\;\sett^m}}\;S_{\intype{w\;}{\;s_2[x,\beta]}}\;\Psi[x,\beta,w]$
be $\overline{\tau}$ and $\overline{\gamma[x]}$ respectively and define
\begin{eqnarray*}
  \overline{\sigma} & :\equiv &\left\{\begin{array}{l} \Sigma_{\intype{\alpha;\beta\;}{\;\sett^{n+m}}} \\
  ~~~~~ S_{\intype{z\;}{\;\Sigma_{\intype{w\;}{\;s_1[\alpha]}}\;s_2[\sigax^{-1}_\tau[\pair(\alpha,w)],\beta]}} \\
  ~~~~~~~~~~\Phi[\alpha,\pi_1(z)] \wedge \Psi[\sigax^{-1}_{\tau[x]}[\pair(\alpha,\pi_1(z))],\beta,\pi_2(z)] \end{array}\right. \\
  \\
  \sigax_\sigma[N] & :\equiv & \left\{\begin{array}{l}\pair(\pi_1(W_1);\pi_1(W_2),\;\pair(\pi_2(W_1),\pi_2(W_2))) \\
  ~~~~W_2 = \sigax_{\tau}[\pi_1(N)] \\
  ~~~~W_2 = \sigax_{\gamma[x]}[\pi_2(N)] \end{array}\right. \\
  \\
  \sigax_\sigma^{-1}[S] & :\equiv & \left\{\begin{array}{l} \pair(\sigax^{-1}_\tau[\pair(\alpha,\pi(\pi_2(S)))],\;\sigax^{-1}_{\gamma[x]}[\pair(\beta,\pi_2(\pi_2(S)))]) \\
    ~~~ \alpha;\beta = \pi_1(S) \end{array}\right.
\end{eqnarray*}

Where in the definition of $\overline{\sigma}$ in the last case we take $\intype{\alpha;\beta}{\sett^{n+m}}$ to abbreviate $\intype{\eta}{\sett^{n+m}}$ and where $\alpha$ is taken to be the first
$n$ sets in $\eta$ and $\beta$ is taken to be to the remaining $m$ sets in $\eta$.  A similar convention applies to the notation $\alpha;\beta = \pi_1(S)$ in the last line of the last case.
\end{proof}

\subsection{Signature Simplification}

It is worth noting the following equations which can be used to simplify signatures.

$$\begin{array}{rclr}
  \Sigma_{\intype{x\;}{\;S_{\intype{z\;}{\;s}} \Phi[z]}}\;S_{\intype{y\;}{\;u}}\;\Psi[x,y] & = & S_{\intype{p\;}{\;s \times u}}\;\Phi[\pi_1(p)]\wedge\Psi[\pi_1(p),\pi_2(p)] & \mbox{for $x \not \in u$} \\
  \\
  \Pi_{\intype{x\;}{\;s}}\;S_{\intype{y\;}{\;u}}\;\Phi[x,y] & = & S_{\intype{f\;}{\;s \rightarrow u}}\;\forall\intype{x}{s}\; \Phi[x,f(x)] & \mbox{for $x \not \in u$}
\end{array}$$

\noindent These equations can be used as rewrite rules. In practice these rules can usually be used to rewrite
a signature-axiom class to one whose signature is a simple type
--- either a set variable, an expressions not containing set variables, a product set $s_1 \times
s_2$ or a function set $s_1\rightarrow s_2$ where $s_1$ and $s_2$ are recursively simple.  However,
this is not always the case.  The definition of a category given in section~\ref{sec:categories}
involves a compatibility requirement on the composition of morphisms.

\subsection{Signature Ambiguity and Signatures with Free Sort Variables}
\label{sec:ambiguous}

Of course the same structure can be contained in multiple classes.  For example, an Abelian group is also a group.  It turns out that the same object can even be assigned
different signatures.  For example we have
\begin{eqnarray*}
  \intype{\alpha}{\sett};\;\intype{x}{\alpha} & \models & \intype{\tuple{\alpha,x}}{(\Sigma_{\intype{\beta\;}{\;\sett}}\;\beta)} \\
    & \models & \intype{\tuple{\alpha,x}}{(\sett \times \alpha).}
\end{eqnarray*}
Different signatures impose different abstract interfaces.  For example we have
\begin{eqnarray*}
  \intype{\alpha}{\sett};\;\intype{x}{\alpha};\; \intype{n}{(\sett \times \alpha)};\;\intype{y}{\alpha} & \models & \intype{(\pi_2(n) = y)}{\bool} \\
  \intype{\alpha}{\sett};\;\intype{x}{\alpha};\; \intype{n}{(\Sigma_{\intype{\beta\;}{\;\sett}}\;\beta)};\;\intype{y}{\alpha} & \not \models & \intype{(\pi_2(n) = y)}{\bool}.
\end{eqnarray*}

We also consider signatures such as
$\Sigma_{\intype{\beta\;}{\;\sett}}\;\beta\rightarrow \alpha$ where $\alpha$ is a sort variable.
Now consider the following valid sequent.
$$\left\{\begin{array}{l}\intype{\alpha}{\sett}; \\\intype{x}{\alpha}; \\\intype{f}{\Sigma_{\intype{\beta\;}{\;\sett}}\;\beta \rightarrow \alpha}; \\\;\intype{i}{\pi_1(f)}\end{array}\right\}
\models \intype{(\;\pi_2(f)(i) = x\;)}{\bool}$$
Here the sort $\alpha$ is ``exposed'' for ``external'' use in the context in which $f$ is defined.
The notion of isomorphism  for the class $\Sigma_{\intype{\beta\;}{\;\sett}}\;\beta\rightarrow \alpha$ only allows bijections on the ``hidden''
sort $\beta$. The isomorphism classes of objects in the class
$\Sigma_{\intype{\beta\;}{\;\sett}}\;\beta\rightarrow \alpha$ correspond to bags
or multisets of values of sort $\alpha$.

\subsection{Points and Representations}
\label{sec:representation}

Cayley's theorem states that every group can be represented by a group of permutations.  This theorem fundamentally involves isomorphism.  More formally it states that
every group is group-isomorphic to a group whose elements are permutations on an underlying set and where the group operation is composition of permutations.
Representation is fundamental to mathematics and we will use Cayley's theorem as an example.

In clarifying representation it will be useful to define $\Gamma \models \intype{e}{\point}$ to mean that the no type can be assigned to $e$ ---
either $\Gamma$ is inconsistent (there are no variable interpretations satisfying $\Gamma$)
or for any value $v$ there exists a variable interpretation $\rho$ satisfying $\Gamma$ with ${\cal V}_\Gamma\double{e}\rho = v$.
For example we have
$$\intype{\alpha}{\sett};\;\intype{x}{\alpha} \models \intype{x}{\point}.$$
If $\Gamma \models \intype{e}{\point}$ then $\Gamma;\Delta \models \intype{e}{\point}$
for any well-formed context extension $\Delta$
because a well-formed context extension must not violate the abstraction barrier on $e$ imposed by $\Gamma$ --- $e$ must be treated as a point in
all well-formed expressions involving $e$.

The elements of a group {\em variable} are points in the sense that $$\intype{G}{\group};\;\intype{x}{\pi_1(G)} \models \intype{x}{\point}.$$
However the elements of a group {\em representation} have structure.
Abbreviating technical details we first define a composition-closed function predicate (CCFPred) by
$$\mathbf{CCFPred} :\equiv \Sigma_{\intype{\alpha\;}{\;\sett}} \;\;\;S_{\intype{P\;}{\;(\alpha \rightarrow \alpha) \rightarrow \bool}}\;\;\;\mathbf{CompClosed}(P)$$
where $\mathbf{CompClosed}(P)$ states that the the predicate $P$, viewed as a set of functions, is closed under composition.
We then have
$$\intype{P}{\mathbf{CCFPred}} \models \intype{\tuple{\predset(\pi_2(P)),\mathbf{Comp}(\pi_2(P))}}{\mathbf{Magma}}$$
Where for $\intype{Q}{\;s \rightarrow \bool}$ we have that $\predset(Q)$ abbreviates $S_{\intype{x\;}{\;s}}\;Q(x)$, and
where for a composition-closed predicate $\intype{Q}{(\alpha \rightarrow \alpha) \rightarrow \bool}$ we have that $\mathbf{Comp}(Q)$ denotes the composition
operation, and where the class $\mathbf{Magma}$ is a signature-axiom class without axioms.
$$\mathbf{Magma} :\equiv \Sigma_{\intype{\alpha\;}{\;\sett}}\;(\alpha \times \alpha) \rightarrow \alpha$$
A group can be defined as a magma satisfying group axioms and we can then write Cayley's theorem as
$$\begin{array}{l}
  \forall\;\intype{G}{\group} \\
  ~~~\exists\;\intype{P}{\mathbf{CCFPred}} \\
  ~~~~~~~~~~G =_{\group} \;\tuple{\predset(\pi_2(P)), \;\mathbf{Comp}(\pi_2(P))}
\end{array}$$
While the elements of a group variable are points, the elements of a permutation group are functions.  However,
a permutation group is still a possible semantic value for a permutation variable.  Being a point represents a lack
of information.

\ignore{
We can also define $\Gamma \models \intype{e}{\mathbf{Pair}}$ to mean that in all variable interpretations $\rho$ satisfying $\Gamma$
we have ${\cal V}_\Gamma\double{e}\rho$ is a value tagged as a pair.
We define $\Gamma \models \intype{e}{\mathbf{Function}}$ to mean that in all variable interpretations $\rho$ satisfying $\Gamma$
we have ${\cal V}_\Gamma\double{e}\rho$ is a value tagged as a function.  The following lemma will be useful in section~\ref{sec:subst}.
}

\section{Isomorphism}
\label{sec:isomorphism}

The sine-qua-non of isomorphism-motivated type theory, is the
inference rule of the substitution of isomorphics.

\begin{infrule}
  \Gamma & & \models & \intype{\tau}{\class} \\
  \Gamma & \intype{x}{\sigma} & \models & \intype{g[x]}{\tau} \\
  \Gamma & & \models & N =_\sigma N' \\
  \hline \\
  \Gamma & &  \models & g[N] =_\tau g[N']
\end{infrule}

Here we will define Bourbaki isomorphism in the context of the dependent type theory defined in section~\ref{sec:constructs}
and prove the soundness of the substitution of isomorphics.  As in colloquial mathematics, a Bourbaki isomorphism is a bijection
between the corresponding sorts of two structures satisfying certain commutativity conditions.  An isomorphism between graphs $G$ and $G'$ is a bijection between the nodes
of $G$ and the nodes of $G'$ that identifies the edges of $G$ with the edges of $G'$.  The isomorphism is the bijection.  For a multi-sorted signature, as in a colored graph,
we will take an isomorphism to be a system of bijections between corresponding sorts.
When treating propositional equality as isomorphism in MLTT we write the set of $\sigma$-isomorphisms from $N$ to $N'$ as $\id(\sigma,N,N')$.
Under the Bourbaki semantics presented here $\id(\sigma,N,N')$ is taken to be meta-notation for a set expression denoting a set of bijections satisfying conditions.
We will define $\id(\sigma,N,N')$ for the case where $\sigma$ is a signature-axiom class.  For a general class $\sigma$ we will define the isomorphism
set $\id(\sigma,N,N')$ through the transformation to a signature-axiom class.
\begin{equation}
  \label{eq:geniso}
  \id(\sigma,N,N') :\equiv \id(\overline{\sigma},\sigax_\sigma(N),\sigax_\sigma[N'])
\end{equation}

The original definition of structure and isomorphism due to Bourbaki, and described in section~\ref{sec:Bourbaki}, involves the notion of transport.
For simple signatures (as opposed to dependent type signatures) the definitions of section~\ref{sec:Bourbaki} straightforwardly generalize to multi-sorted classes.
For a signature-axiom class $\Sigma_{\intype{\alpha\;}{\;\sett}}\;S_{\intype{x\;}{\;s[\alpha]}}\;\Phi[\alpha,x]$ and for $\intype{A,A'}{\sett^n}$ we have that
two structures
$$\intype{\tuple{A,x},\;\tuple{A',x'}}{\;\;\Sigma_{\intype{\alpha\;}{\;\sett^n}}\;S_{\intype{x\;}{\;s[\alpha]}}\;\Phi[\alpha,x]}$$
are isomorphic if there exists an $n$-tuple $f$ of bijections with $\intype{f_i}{A_i \rightarrow A'_i}$ such that $f$
carries $x$ to $x'$.   In section~\ref{sec:Bourbaki} transport is written as $f_s(x) = x'$.
When treating propositional equality as isomorphism in MLTT
the notation $f_s$ of section~\ref{sec:Bourbaki} is written as
$\subst(\sett^n,A,A',f,s)$ 
where $s$ is a macro such that $s[\alpha]$ denotes a set.
We then have
\begin{eqnarray}
  & & \id(\;(\Sigma_{\intype{\alpha\;}{\;\sett^n}}\;S_{\intype{x\;}{\;s[\alpha]}}\;\Phi[\alpha,x]),\;\tuple{A,X},\;\tuple{A',X'}) \nonumber \\
  \label{eq:id}
  & :\equiv & S_{\intype{f\;}{\;\id(\sett^n,A,A')}}\;\;\subst(\sett^n,A,A',f,s)(X) = X'  
\end{eqnarray}
where $\subst(\sett^n,A,A',f,s)$ is defined below in a way that allows it to be represented by a lambda expression.
We can take $N =_\sigma N'$ to be a formula stating that there exists an element of $\id(\sigma,N,N')$.
In BTT no linguistic extensions are required to handle isomorphism --- there is no need for a J operator.

\subsection{Defining {\bf Subst} and {\bf Transport}}

As in section~\ref{sec:Bourbaki}, we will define $\subst(\sett^n,A,A',f,s)$ by case analysis on $s$.

\begin{lemma}[Set Case Analysis]
  \label{lem:cases}
  If $\Gamma \models \intype{s}{\sett}$ then either
  \begin{itemize}
  \item $\Gamma;\;\intype{x}{s} \models \intype{x}{\bool}$,
  \item $\Gamma;\; \intype{x}{s} \models \intype{x}{\point}$ with $\Gamma;\;\intype{x}{s} \models \intype{x}{\pi_{i_1}(\cdots\pi_{i_n}(m)\cdot)}$ where $\intype{m}{\sigma} \in \Gamma$,
  \item $\Gamma;\;\intype{x}{s} \models \intype{x}{\Sigma_{\intype{x\;}{\;u}} w[x]}$ for $\Sigma_{\intype{x\;}{\;u}} w[x]$ no larger than $s$,
  \item or $\Gamma;\;\intype{x}{s} \models \intype{x}{\Pi_{\intype{x\;}{\;u}} w[x]}$ for  $\Pi_{\intype{x\;}{\;u}} w[x]$ no larger than $s$.
  \end{itemize}
  where in the second clause we allow $\intype{x}{m}$ for the case where $n=0$.
\end{lemma}

\begin{proof}
The proof is by induction on the construction of $s$.
We will first list some ways set expressions {\em cannot} be formed.
Note that by the induction hypothesis for a previously constructed set $s$
we have that if $\Gamma \models \intype{e}{s}\!:\!\sett$ then $e$ is either a Boolean, a point, a pair or a function.
This implies that for a previously constructed set $s$ with $\Gamma \models \intype{e}{s}\!:\!\sett$ we have $\Gamma \not \models \intype{e}{\sett}$.
Hence a newly constructed set expression cannot itself be a set element.  This implies that a newly constructed set expression cannot be a function
application or a projection $\pi_i(p)$ where we have $\intype{p}{s}\!:\!\sett$.
a variable declared in a context either denotes a set element or a class element. Set elements cannot denote sets.

We now consider the ways in which set expressions {\em can} be formed.
For a subset expression $S_{\intype{x\;}{\;\tau}}\;\gamma[x]$ the lemma follows from the
induction hypothesis on $\tau$.
If $\Gamma \models \intype{\alpha}{\sett}$ with $\alpha\infc \sigma \in \Gamma$ then $\sigma$ cannot be a pair class and by induction on $\sigma$
one can show that $\Gamma;\;\intype{x}{\sigma} \models x\infc \point$.  So this case is covered by the first clause above.
Set-level pair types and function types correspond to the
last two cases of the lemma.

The final case is where $\Gamma \models \intype{\pi_i(p)}{\sett}$ for some projection expression $\pi_i(p)$
with $\Gamma \models \intype{p}{\sigma}\!:\!\class$.  Here $\sigma$ must be a class of pairs.
A class expression is either the constant $\sett$, a pair class
expression $\Sigma_{\intype{x\;}{\;\tau}}\;\gamma[x]$ where either $\tau$ or $\gamma[x]$ are class expressions, or a subclass expression $S_{\intype{x\;}{\;\tau}}\;\Phi[x]$.
By induction on class expressions, the members of a class expression are either sets, or are pairs that contains sets within them.
By the above comments, if $p$ is a previously constructed pair expression with $\Gamma \models \intype{p}{s}\!:\!\sett$ then $p$
cannot contain a set within it.
This implies that pairs in classes are disjoint from pairs in sets.  Hence for $\Gamma \models \intype{p}{\sigma}\infc\class$ we have that
$p$ cannot be an application expression and hence $p$ must a pair expression or a projection expression.
We now show by induction on the size of $p$ that for $\Gamma\models \intype{\pi_i(p)}{\sett}$ we have $\Gamma \models \pi_i(p) = \pi_{i_1}(\cdots\pi_{i_n}(n)\cdot)$
for $n\infc \sigma \in \Gamma$ with $\Gamma \models \sigma\infc \class$.
This is immediate if $p$ does not contain any pair expressions.
If $\pi_i(p)$ does contain a pair expression then it is either equal to a previously constructed expression, in which case the theorem follows from the induction hypothesis,
or is equal to a shorter projection expression in which case the lemma follows from the induction on the size of the projection expressions.
\end{proof}

When $s[\alpha]$ is a simple type, as in section~\ref{sec:Bourbaki}, the definition of $\subst$ is a straightforward
multi-sorted generalization of the definition in section~\ref{sec:Bourbaki}.
However, here we must handle the case where $s[\alpha]$ is a dependent type. If $s[\alpha]$ is a dependent pair type then the type of the second component
of a pair depends on the value of the first component.  To define the transport operation on the second component we need to know the value of the first component.
To handle this we introduce a transport function $\trans(\sett^n,A,A,f,u,X,s)$.  This transport operation can transport the second component of a pair from a dependent pair type where
the $X$ contains the value of the first component.
More specifically, for

\begin{infrule}
  \Gamma & & \models & \intype{A,A'}{\sett^n} \\
  \Gamma & & \models & \intype{f}{\id(\sett^n,A,A')} \\
  \Gamma; & \intype{\alpha}{\sett^n} & \models & \intype{u[\alpha]}{\sett} \\
  \Gamma & & \models & \intype{X}{u[A]} \\
  \Gamma; & \intype{\alpha}{\sett^n};\;\intype{x}{u[\alpha]} & \models & \intype{s[\alpha,x]}\sett
\end{infrule}

we will define $\trans(\sett^n,A,A',f,u,X,s)$ where
\begin{eqnarray*}
  & & \subst(\sett^n,A,A',f,s) \\
  & = & \trans(\sett^n,A,A',f,(\lambda\;\intype{\alpha}{\sett^n} \;\bool),\true,(\lambda\;\intype{\alpha}{\sett^n}\;\lambda \intype{P}{\bool}\;s[\alpha]))
\end{eqnarray*}
In the following definition we omit the first three arguments to $\subst$ and $\trans$ to shorten the notation.

\begin{itemize}
  \item If $\alpha$ does not occur in $s[\alpha,x]$ then
    $$\trans(f,u,X,s)(y) = y.$$
  \item If $\Gamma;\; \intype{\alpha}{\sett^n};\;\intype{x}{u[\alpha]};\;\intype{y}{s[\alpha,x]} \models \intype{y}{\pi_{i_1}(\cdots\pi_{i_n}(\alpha)\cdot)}$ then
        $$\trans(f,u,X,s)(y) = \pi_{i_1}(\cdots\pi_{i_n}(f)\cdot)(y).$$
  \item If $\Gamma;\; \intype{\alpha}{\sett^n};\;\intype{x}{u[\alpha]};\;\intype{y}{s[\alpha,x]} \models \intype{y}{\Sigma_{\intype{z\;}{\;v[\alpha,x]}}\;w[\alpha,x,z]}$
    \begin{eqnarray*}
      \trans(f,u,X,s)(y) &= &\tuple{\begin{array}{l} \trans(f,u,X,v)(\pi_1(y)), \\ \trans(f,\Sigma_{\intype{x\;}{\;s[\alpha]}}\;v[\alpha,x],\tuple{X,\pi_1(y)},\tilde{w})(\pi_2(y))\end{array}} \\
      \tilde{w}[\alpha,\tuple{x,z}] & = & w[\alpha,x,z]
      \end{eqnarray*}
  \item If $\Gamma;\; \intype{\alpha}{\sett^n};\;\intype{x}{u[\alpha]};\;\intype{y}{s[\alpha,x]} \models \intype{y}{\Pi_{\intype{z\;}{\;v[\alpha,x]}}\;w[\alpha,x,z]}$ then
    \begin{eqnarray*}
      \trans(f,u,X,s)(y) & = & \lambda\;\intype{z'}{v[A',X']}\;\;\trans(f,\Sigma_{\intype{x\;}{\;s[\alpha]}}\;v[\alpha,x],\tuple{X,z},\tilde{w})(y(z)) \\
      X' & = & \subst(f,u)(X)\\
      z & = & \trans^{-1}(f,u,X,v)(z') \\
      \tilde{w}[\alpha,\tuple{x,z}] & = & w[\alpha,x,z]
\end{eqnarray*}

\end{itemize}

This recursive definition of $\trans(f,u,X,s)$ is well founded --- each recursive call reduces the size of $s$ while maintaining the sum of the sizes of $u$ and $s$
or reduces the sum of the sizes of $u$ and $s$ by eliminating $s$ in the call to $\subst(f,u)$ in the last clause.
The definition provides an explicit representation of $\trans(f,u,X,s)$ as a lambda expression.
It is possible to show by induction on the definition that $\trans(f,u,X,s)$ is a bijection from $s[A,X]$ to $s[A',X']$ where $X' = \subst(f,u)(X)$.
Or written another way we have
\begin{equation}
  \label{eq:transport}
  \Gamma \models \intype{\trans(\sett^n,A,A',f,u,X,s)}{\Bi[s[A,X],s[A',\subst(f,u)(X)]]}
\end{equation}
In particular 
$$\Gamma \models \intype{\subst(\sett^n,A,A'f,s)}{\Bi[s[A],s[A']]}.$$
The definition of $\subst(\sett^n,A,A',f,s)$ now completes the definition of the isomorphism set $\id(\sigma,N,N')$ as given in equations (\ref{eq:geniso}) and
~(\ref{eq:id}).

\subsection{The Soundness of Substitution}

We finally prove the soundness of substitution.
We assume $\Gamma;\;\intype{n}{\sigma} \models \intype{g[n]}{\tau}$
and
$\Gamma \models \intype{f_\sigma}{\id(\sigma,N,N')}$ and construct $f_\tau$ with $\Gamma \models \intype{f_\tau}{\id(\tau,g[N],g[N'])}$.
It suffices to consider the case where $\sigma$ and $\tau$ are both signature-axiom classes.
\begin{eqnarray}
  \label{sigmadef}
  \sigma & = & \Sigma_{\intype{\alpha\;}{\;\sett^n}}\;S_{\intype{x\;}{\;s[\alpha]}}\;\Phi[\alpha,x] \\
  \label{taudef}
  \tau & = & \Sigma_{\intype{\beta\;}{\;\sett^m}}\;S_{\intype{y\;}{\;u[\beta]}}\;\Psi[\beta,y]
\end{eqnarray}
We can write $N$ and $N'$ as $\tuple{A,X}$ and $\tuple{A',X'}$
and write $g[N]$ and $g[N']$ as $\tuple{B[\tuple{A,X}],Y[\tuple{A,X}]}$ and $\tuple{B[\tuple{A',X'}],Y[\tuple{A',X'}]}$.
The sequent $\Gamma;\;\intype{n}{\sigma} \models \intype{g[n]}{\tau}$
implies
$$\Gamma;\;\intype{\alpha}{\sett^n}\;\intype{x}{(S_{\intype{x\;}{\;s[\alpha]}}\;\Phi[\alpha,x])} \models \intype{B[\tuple{\alpha,x}]}{\sett^m}$$
By the definition of $\id(\sigma,\tuple{A,X},\tuple{A',X'})$ we have that $f_\sigma$ is a tuple of bijections
with $\intype{\pi^i(f)}{\Bi[\pi^i(A),\pi^i(A')]}$ where $\pi^i$ selects the $i$th element of a sequence.
Defining $\eta$ to be the signature-axiom set in (\ref{sigmadef})
$$\eta = \lambda\; \intype{\alpha}{\sett^n}\;S_{\intype{x\;}{\;s[\alpha]}}\;\Phi[\alpha,x]$$
we have
$$\Gamma \models \intype{\trans(f,\eta,X,(\lambda\;\intype{\alpha}{\sett^n}\;\lambda\;\intype{x}{\eta[\alpha]}\;\pi^i(B[\alpha,x])))}{\Bi[\pi^i(B[A,X]),\;\pi^i(B[A',X'])]}$$
\noindent We now take $f_\tau$ to be the tuple of bijections defined by
$$\pi^i(f_\tau) = \trans(f_\sigma,\eta,X,(\lambda\;\intype{\alpha}{\sett^n}\;\lambda\;\intype{x}{\eta[\alpha]}\;\pi^i(B[\alpha,x]))).$$
We must now show
$$\Gamma \models \intype{f_\tau}{\id(\tau,g(N),g(N'))}$$
which is equivalent to showing
$$\subst(\sett^m,B[A,X],B[A',X'],f_\tau,u)(Y[A,X]) = Y[A',X']$$
But this is implied by
\begin{eqnarray*}
  & & \subst(\sett^m,B[A,X],B[A',X'],f_\tau,s_\tau) \\
  & = &\trans(\sett^n,A,A',f_\sigma,u,X,(\lambda\;\intype{\alpha}{\sett}\;\lambda\;\intype{x}{u[\alpha]}\;s_\tau[B[\alpha,x]]))
\end{eqnarray*}
which can be proved by induction on the size of $s_\tau$.
  
\subsection{BTT as a Model of MLTT: The Semantics of J}
\label{sec:meets}

So far we have not considered axiom J of Marten-L\"{o}ff type theory (MLTT).  In fact we do not see any role for $J$
in the formulation of Bourbaki type theory.  However, it seems worth noting that
the operation $J$ can be defined in the Bourbaki model in such a way that the $J$ axiom holds.
Under this semantic definition of $J$ the Bourbaki model is a model of a version MLTT.  The inference rule for $J$ can be written as

\noindent J
\begin{infrule}
    \Gamma; & \intype{n,n'}{\;\sigma};\;\intype{g}{\id(\sigma,n,n')} &\models& \intype{\tau[n,n',g]}{U_i} \\
    \Gamma & & \models& \intype{f}{\id(\sigma,N,N')} \\
    \Gamma; & \intype {n}{\sigma}  & \models &\intype{\delta[n]}{\tau[N,N,\mathbf{Refl}(\sigma,N)]} \\
    \hline \\
    \Gamma & & \models & \intype{J(\sigma,N,N',f,\tau,\delta)}{\;\tau[N,N',f]}
\end{infrule}

\noindent We will give a semantics for an operator $J'$ which satisfies the stronger rule

\noindent J'
\begin{infrule}
    \Gamma; & \intype{n,n'}{\;\sigma};\;\intype{g}{\id(\sigma,n,n')} &\models& \intype{\tau[n,n',g]}{U_i} \\
    \Gamma & & \models& \intype{f}{\id(\sigma,N,N')} \\
    \hline \\
    \Gamma & & \models & \intype{J'(\sigma,N,N',f,\tau)}{\;\;\tau[N,N,\refl(\sigma,N)] \rightarrow \tau[N,N',f]}
\end{infrule}

\noindent We can then define $J$ by
$$J(\sigma,N,N',f,\tau,\delta) = J'(\sigma,N,N',f,\tau)(\delta(N))$$

First we reduce defining $J'$ to the case where $\sigma$ is a signature-axiom class using the following reduction
where we note that conversion to a signature-axiom class does not change the isomorphism $f$ --- isomorphisms are
defined anyway by first converting to a signature-axiom class.
\begin{eqnarray*}
  J'(\sigma,N,N',f,\tau) & = & J'(\overline{\sigma},\sigax_\sigma(N),\sigax_\sigma(N'),f,\tilde{\tau}) \\
  \tilde{\tau}[n,n',g] & = & \tau[\sigax_\sigma^{-1}(n),\sigax_\sigma^{-1}(n'),g]
\end{eqnarray*}
Now we consider the case where $\sigma$ is a signature-axiom class.
$$\sigma = \Sigma_{\alpha:\sett^n}\;S_{x:u[\alpha]}\;\Phi[\alpha,x]$$
Under this assumption we can define the semantics of $J'$
\begin{eqnarray*}
J'(\sigma,\tuple{A,X},\tuple{A',X'},f,\tau) & :\equiv &\trans(\sett^n,A,A',f,\eta,\tuple{X,\mathbf{Refl}(\sigma,N)},\gamma) \\
\eta[\alpha] & :\equiv & \Sigma_{x:(S_{x:u[\alpha]}\;\Phi[\alpha,x])} \; \id(\sigma, \tuple{A,X},\tuple{\alpha,x}) \\
\gamma[\alpha,\tuple{x,g}] &:\equiv & \tau[\tuple{A,X},\tuple{\alpha,x},g]
\end{eqnarray*}
For this to be well defined for the general case of $J'$ we need to generalize $\trans(\sett^n,A,A',f,u,X,s)$
to handle the case where $u[\alpha]$ and $s[\alpha,x]$ could be classes.  This is done by generalizing lemma~\ref{lem:cases} to
handle $\Gamma \models \intype{s}{U_i}$ rather than just $\Gamma \models \intype{s}{\sett}$ and where the proof is extended to handle the case of $\Gamma \models \intype{s}{\class}$.
The definition of $\trans(\sett^n,A,A',f,u,X,s)$ then remains unchanged involving the same cases on $s$.
We omit further details.

\noindent Equation (\ref{eq:transport}) also holds under the more general version of $\trans$ and we have
$$\intype{J'(\sigma,\tuple{A,X},\tuple{A',X'},f,\tau)(\tuple{A,X})}{\;\;\tau[\tuple{A,X},\tuple{A',X'},\subst(f,\eta)(\refl(\sigma,N))]}$$
$$\eta[\tuple{\alpha,x}]  =  \id(\sigma,\tuple{A,X},\tuple{\alpha,x})$$
For this $\eta$ we have
$$\subst(f,\eta)(\refl(\sigma,N)) = f$$

\noindent Finally we note that $\trans$, and hence $\subst$, can be expressed in terms of $J'$ as
\begin{eqnarray*}
  \trans(\sett^n,A,A',f,u,X,s) & = & J'((\Sigma_{\alpha:\sett^n}\;u[\alpha]),\tuple{A,X},\tuple{A',X'},f,\tau) \\
  \tau[\tuple{\alpha,x},\tuple{\beta,y},g] & = & s[\beta,y]
\end{eqnarray*}
We leave it to the reader to verify that $J'$ can be expressed in terms of $J$.

\section{Limitations}

The dependent type theory developed here has some limitations that are best exhibited by considering categories and sheaves.
It is easy to represent the class of small categories as a signature-axiom class in the version of Bourbaki type theory presented here.
However, the objects in the category of topological spaces --- the topologies --- form a proper class.
The Bourbaki type theory defined here does not support signature-axiom classes whose sorts are proper classes.
Sheaves raise a related issue.  The definition of a sheaf fails to be a signature-axiom class because the data of a sheaf over a topological space
$X$ involves a map $\intype{f}{\mathbf{OpenSet}(X) \rightarrow \sett}$.  In the system defined here $f$ must be a macro --- we have required that
semantic functions must be from sets to sets and a semantic function from a set to the proper class of all sets is not allowed.  So the class of sheaves
over a topological space $X$ is not a signature-axiom class.

Both of these issues can be handled in a version of the groupoid model with universes \cite{GRPD}.
However, moving to the groupoid model looses the advantages of Bourbaki type theory.
We loose the intuitive representation of isomorphism as bijections between sorts
and the straightforward approach to a constructive proof (providing computational content) for the validity of the substitution of isomorphics.
We take it to be an open problem to expand Bourbaki type theory to the full power of the groupoid model while
preserving these advantages in some form.

\section{Conclusions}

Isomorphism is central to both human mathematical thought and to the structure of mathematics.
Bourbaki type theory is intended to provide a formal treatment of isomorphism in correspondence with
human thought and the structure of mathematics.  Presumably people recognize isomorphism as equality
because the language of mathematics has a grammar supporting the validity of the substitution of
isomorphics.

It is not expected that studying dependent type theory will improve the ability of
mathematicians to do mathematics.  Rather, it seems clear that the understanding of types (concepts)
and isomorphism is already subconsciously ingrained into human thought.  This is analogous to the
grammar of natural languages, such as English, where native speakers speak grammatically and can
recognize ungrammatical sentences even when they have no ability to enumerate rules of grammar.
Speaking and understanding language would be impossible if one had to think consciously about all
the rules being used at a subconscious level.  Still, the study of language and grammar seems
interesting as a scientific investigation in it own right.

The study of grammar, and the role of grammar in thought, seems most significant from the
perspective of artificial intelligence.  If indeed effective human mathematical thought rests on grammatical
properties of the language of thought, this should be relevant to the construction of automated
reasoning systems.  One might expect that the grammatical structure of mathematical thought would be
related to thought generally and perhaps even to common sense reasoning.  One should note, however, that
current trends in AI have replaced logic with deep networks as the central paradigm.  Presumably logic and deep networks
can be combined in so-called neuro-symbolic systems.  Dependent type theory seems likely to be relevant to this endeavor.

%\bibliographystyle{apalike}
%\bibliography{mathzero}

\begin{thebibliography}{}

\bibitem[Birkhoff, 1967]{Birk}
Birkhoff, G. (1967).
\newblock {\em Lattice Theory, 3rd edition}.
\newblock American Mathematical Society.

\bibitem[Bourbaki, 1939]{Bourbaki39}
Bourbaki, N. (1939).
\newblock {\em Th\'{e}orie des Ensembles}.

\bibitem[de~Moura et~al., 2015]{LEAN}
de~Moura, L., Kong, S., Avigad, J., van Doorn, F., and von Raumer, J. (2015).
\newblock The lean theorem prover (system description).
\newblock In {\em 25th International Conference on Automated Deduction (CADE)}.
\newblock see https://leanprover-community.github.io/.

\bibitem[Hofmann and Streicher, 1998]{GRPD}
Hofmann, M. and Streicher, T. (1998).
\newblock The groupoid interpretation of type theory.
\newblock In {\em Twenty-five years of constructive type theory ({V}enice,
  1995)}, volume~36 of {\em Oxford Logic Guides}, pages 83--111. Oxford Univ.
  Press, New York.

\bibitem[HoTT-Authors, 2013]{Hott}
HoTT-Authors (2013).
\newblock Homotopy type theory, univalent foundations of mathematics.
\newblock
  http://hottheory.files.wordpress.com/2013/03/hott-online-611-ga1a258c.pdf.

\bibitem[Kapulkin et~al., 2012]{Simpsets}
Kapulkin, C., Lumsdaine, P.~L., and Voevodsky, V. (2012).
\newblock The simpicial model of univalent foundations.
\newblock {\em CoRR}, abs/1211.2851.

\bibitem[Lee and Werner, 2011]{LeeWerner11}
Lee, G. and Werner, B. (2011).
\newblock Proof-irrelevant model of cc with predicative induction and
  judgmental equality.
\newblock {\em Logical Methods in Computer Science}, 7.

\bibitem[Martin-L\"{o}ff, 1973]{mltt1973}
Martin-L\"{o}ff, P. (1973).
\newblock An intuitionistic theory of types: predicative part.
\newblock In {\em Logic Colloquium '73, volume~80 of Studies in Logic and the
  Foundations of Mathematics}. North-Holland.

\bibitem[Miquel and Werner, 2003]{MiquelWerner03}
Miquel, A. and Werner, B. (2003).
\newblock The not so simple proof-irrelevant model of cc.
\newblock In {\em Types for proofs and programs}, volume 2646 of {\em Lecture
  Notes in Comput. Sci.}, pages 240--258. Springer.

\bibitem[N\"{o}strom, 1990]{mltt1990}
N\"{o}strom (1990).
\newblock {\em Programming in Martin L\"{o}ff's Type Theory}.
\newblock Oxford University Press.

\bibitem[Reynolds, 1984]{Reynolds84}
Reynolds, J.~C. (1984).
\newblock Polymorphism is not set-theoretic.
\newblock In {\em Semantics of Data Types}, volume 173 of {\em Lecture Notes in
  Comput. Sci.}, page 145–156. Springer.

\bibitem[Rota, 1997]{Rota}
Rota, G. (1997).
\newblock {\em Indiscrete Thoughts}.
\newblock Birkauser Boston, Inc.

\bibitem[van Doorn et~al., 2020]{MathLib}
van Doorn, F., Ebner, G., and Lewis, R. (2020).
\newblock Maintaining a library of formal mathematics.
\newblock In {\em Intelligent Computer Mathematics (CICM)}.

\end{thebibliography}

\end{document}